\documentclass[runningheads]{llncs}

\usepackage{graphicx}
\usepackage[colorlinks]{hyperref}

\usepackage{afterpage}
\usepackage{amsmath}
\usepackage{amsfonts}
\usepackage{amssymb}
\usepackage[switch]{lineno}

\newcommand{\cands}{\mathcal{C}}
\newcommand{\election}{\mathcal{E}}
\newcommand{\ballots}{\mathcal{B}}
\newcommand{\threshold}{\ensuremath{\tau}}
\newcommand{\unmentioned}{\mathcal{U}}
\newcommand{\viable}{\mathcal{V}}
\newcommand{\winners}{\mathcal{W}}
\newcommand{\losers}{\mathcal{L}}

\newcommand{\asserts}{\mathcal{A}}
\newcommand{\frontier}{F}

\usepackage[backgroundcolor=white]{todonotes}

\newcommand{\ignore}[1]{}

\newcommand{\Viable}{\mbox{\emph{Viable}}}
\newcommand{\NonViable}{\mbox{\emph{NonViable}}}

\title{Auditing Hamiltonian Elections}

\author{
Michelle Blom    \inst{1}   \orcidID{0000-0002-0459-9917}  \and
Philip B. Stark  \inst{2}   \orcidID{0000-0002-3771-9604}  \and
Peter J. Stuckey \inst{3}   \orcidID{0000-0003-2186-0459}  \and
Vanessa Teague   \inst{4}   \orcidID{0000-0003-2648-2565}  \and
Damjan Vukcevic  \inst{5,6} \orcidID{0000-0001-7780-9586}}

\authorrunning{Blom M, Stark PB, Stuckey PJ, Teague V, Vukcevic D}

\institute{
School of Computing and Information Systems, University of Melbourne,
Parkville, Australia \\
\email{michelle.blom@unimelb.edu.au}
\and
Department of Statistics, University of California, Berkeley, USA
\and
Department of Data Science and AI, Monash University, Clayton, Australia
\and
Thinking Cybersecurity Pty. Ltd.
\and
School of Mathematics and Statistics, University of Melbourne, Parkville,
Australia
\and
Melbourne Integrative Genomics, University of Melbourne, Parkville, Australia}
\date{}

\begin{document}
%\linenumbers
\maketitle

\begin{abstract}
Presidential primaries are a critical part of the United States Presidential
electoral process, since they are used to select the candidates in the
Presidential election.  While methods differ by state and party, many primaries
involve proportional delegate allocation using the so-called Hamilton method.
In this paper we show how to conduct risk-limiting audits for delegate
allocation elections using variants of the Hamilton method where the viability
of candidates is determined either by a plurality vote or using instant runoff
voting.  Experiments on real-world elections show that we can audit primary
elections to high confidence (small risk limits) usually at low cost.
\end{abstract}

\section{Introduction}

Presidential primary elections are a critical part of the United States
electoral process, since they are used to  select the final candidates
contesting the Presidential election for each of the major parties.   For that
reason it is important that the result of these primaries be trustworthy.
While the method used for primary elections differs by party and state, the
majority of such elections use delegate allocation by proportional
representation, the so-called Hamilton method, named after its inventor,
Alexander Hamilton.

Risk-limiting audits (RLAs)~\cite{stark2008conservative} require a durable,
trustworthy record of the votes, typically paper ballots marked by hand, kept
demonstrably secure.  RLAs end in one of two ways: either they produce strong
evidence that the reported winners really won, or they result in a full manual
tabulation of the paper records.  If a RLA leads to a full manual tabulation,
the outcome of the tabulation replaces the original reported outcome if they
differ, thus correcting the reported outcome (if the paper trail is
trustworthy).  The probability that a RLA fails to correct a reported outcome
that is incorrect before that outcome becomes official is bounded by a ``risk
limit.'' An RLA with a risk limit of 1\%, for example, has at most a 1\% chance
of failing to correct a reported election outcome that is wrong; equivalently,
it has at least a 99\% chance of correcting the reported outcome if it is
wrong.  RLAs are becoming the de-facto standard for post-election audits.  They
are required by statute in Colorado, Nevada, Rhode Island, and
Virginia,\footnote{Virginia's audit does not take place until after the outcome
is certified, so it cannot limit the risk that an incorrect reported outcome
will become final: technically, it is not a RLA.} for some government elections
(not primaries which are party elections), and have been piloted in over a
dozen US states and Denmark.  They are recommended by the US National Academies
of Science, Engineering, and Medicine and endorsed by the American Statistical
Association.  Risk-limiting audits of limited scope have begun to be applied to
US primary elections; our methods here would allow RLAs of the full elections.

In this paper we describe the first method that we are aware of for conducting
an RLA for delegate allocation by proportional representation elections, which
we call \emph{Hamiltonian elections}.  In addition to primary elections in some
states in the USA, this type of election is used in Russia, Ukraine, Tunisia,
Taiwan, Namibia and Hong Kong.  We do so by adapting auditing methods designed
for plurality and instant runoff voting (IRV) elections for auditing the
viability of candidates, and generating a new kind of audit for proportional
allocation.

A delegate allocation election by proportional representation is a complex form
of election.  Rather than simply electing candidates, the result of the
election is to assign some number of delegates to some of the candidates.  In
the first stage of the election, the process determines the subset of
candidates that are eligible or \textit{viable} (for Democratic primaries,
candidates need to receive at least 15\% of the vote).  In the second step,
delegates are awarded to these viable candidates in approximate proportion to
their vote.  An RLA must determine the correctness of both the set of viable
candidates and the number of delegates assigned to each viable candidate.

The first stage of the election uses either simple plurality voting, where each
ballot is a vote for at most one candidate, or IRV, where each ballot is a
ranking of some or all candidates.  In IRV, candidates with the fewest
first-choice ranks are eliminated and each ballot that ranked them first is
reassigned to the next most-preferred ranked candidate on that ballot. 

There is considerable work on both comparison audits and ballot-polling audits
for plurality elections~\cite{lindemanStark12,shangrla}, but few for more
complex election types.  Sarwate {\it et al.}~\cite{sarwate2013risk} consider
IRV and some other preferential elections.  Kroll {\it et al.}~\cite{kroll2014}
show how to audit the overall US electoral college outcome, but not the
allocation of individual delegates.  Stark and Teague~\cite{StarkTeague2014}
devise audits for the D'Hondt method for proportional representation, which is
related to but distinct from Hamiltonian methods.  Blom {\it et
al.}~\cite{blom18,blom16} describe efficient audits for IRV.  As far as we
know, there is no other auditing method for Hamiltonian Elections, nor any that
combines a proportional representation method with IRV.

\section{Hamiltonian Elections}
\label{sec:democratic}

We have a set of $n$ candidates $\cands$, a set of cast ballots\footnote{We do
not distinguish between ballots and ballot cards; in general, ballots consist
of one or more cards, of which at most one contains any given contest.}
$\ballots$, and a number of delegates $D$ to be awarded to the candidates based
on the votes.  The \emph{Hamilton} or \emph{largest remainder} method, invented
by Alexander Hamilton in 1792, allocates the delegates in approximate
proportion to the votes the candidates received.

In a \emph{pure Hamiltonian election}, also known as the \emph{Hamilton
method}, delegates are directly allocated based on the proportion of the vote.
But most delegate elections use some form of \emph{exclusion} of some
candidates before the delegates are apportioned.

A \emph{Hamiltonian election with exclusion} first determines which candidates
in $\cands$ are \textit{viable}---eligible to be awarded one or more delegates.
Typically, exclusion involves a plurality vote.  Each ballot is a vote for at
most one candidate.  If a candidate receives a threshold proportion
$\threshold$ of the valid votes, the candidate is considered
viable.\footnote{There are more complicated alternate rules for the case where
no candidate reaches $\threshold$; we do not consider this case here.} The
votes cast for viable candidates are referred to as \textit{qualified votes}.
The qualified votes are used to allocate delegates, as described later in this
section.

\begin{example}\label{ex:plural}
Consider an example Hamiltonian election with exclusion with 4~candidates, Ann,
Bob, Cal, and Dee and a viability threshold of $\threshold = 15\%$.
Figure~\ref{fig:plural}(a) shows the tally of votes for each candidate, and the
percentage of the overall vote that each candidate received.  Ann and Bob
received more than 15\% of the vote and are viable candidates.
\begin{figure*}[t]
\centering
\begin{tabular}{c@{~~~}c}
\begin{tabular}{|l|rr|}
\hline
Candidate & Votes & Proportion \\ \hline
Ann & 57,532 & 76.1\% \\
Bob & 15,630 & 20.6\% \\
Cal & 1,600  & 2.1\% \\
Dee & 846 & 1.1\% \\
\hline
Total Votes & \bf 75,608 & \bf 100.0\% \\
\hline
\end{tabular}
&
\begin{tabular}{|l|rr|}
\hline
Candidate & Votes & Proportion \\ \hline
Ann & 57,532 & 78.6\% \\
Bob & 15,630 & 21.4\% \\
\hline
Qualified Votes & \bf 73,162 & \bf 100.0\% \\
\hline
\end{tabular} \\
(a) & (b)
\end{tabular}
\caption{(a) Votes and (b) Qualified Votes in a Hamiltonian election with
plurality-based exclusion.\label{fig:plural}}
\end{figure*}
\end{example}

For elections with many candidates, a plurality exclusion might eliminate all
of them.  In an \emph{instant-runoff Hamiltonian election} the viable
candidates are determined by a form of IRV.  Each ballot is now a partial
ranking of the candidates, and the viable candidates are determined as follows:
\begin{enumerate}
\item Initialize the set of candidates.  Each ballot is put in the pile for the
    candidate ranked highest on that ballot. 
\item If every (remaining) candidate has $\geqslant \threshold$ of the votes in
    their pile, we finish the candidate selection process.  All of these
    remaining candidates are viable.
\item Otherwise, the candidate with the lowest tally (fewest ballots in their
    pile) is eliminated, and each of their ballots is moved to the pile of the
    next ranked remaining candidate on the ballot. A ballot is \emph{exhausted}
    if all further candidates mentioned on the ballot have already been
    eliminated.
\item We then return to step~2.
\end{enumerate}

\begin{example}\label{ex:irv}
Consider an instant-runoff Hamiltonian election with the same four candidates
as Example~\ref{ex:plural}, the same threshold, and
50,000 ballots with ranking [A,D,C,B] (that is, Ann followed by Dee, then Cal,
       then Bob),
9,630 of [B,C],
6,000 of [C,B],
1,600 of [C],
7,532 of [D,A,C],
and 846 of [D,C].
The IRV election proceeds as follows.  In the first round Cal has the lowest
tally, 7,600 votes, which is 10.052\% of the total vote, and hence less than
15\%.  Cal is eliminated: the 6,000 ballots [C,B] are transferred to Bob, and
the 1,600 ballots [C] are exhausted (removed from consideration).  In the next
round Dee has the lowest tally, 8,378 votes which is 11.080\%, so Dee is
eliminated.  The 7,632 [D,A,C] ballots  are transferred to Ann, and the
remaining 836 [D,C] ballots are exhausted.  In the final round, Bob has the
lowest tally, 20.672\% of the vote, and the process ends since this is greater
than 15\%.  The election is summarized in Figure~\ref{fig:irv}.
\begin{figure*}
\begin{tabular}{|l|rrr|rrr|rrr|}
\hline
& \multicolumn{3}{c|}{Round 1} & \multicolumn{3}{c|}{Round 2} &
\multicolumn{3}{c|}{Final Result} \\
Cand. & Ballot & Number & Prop. & Ballot & Number & Prop. & Ballot & Number & Prop. \\ \hline
Ann & [\textbf{A},D,C,B] & 50,000 & 66.1\% & [\textbf{A},D,C,B] & 50,000 & 66.1\% & [\textbf{A},D,C,B] & 50,000 &  \\
    &  & & & & & & [D,\textbf{A},C] & 7,532 & 76.1\% \\
    \hline
Bob & [\textbf{B},C] & 9,630 & 12.7\% & [\textbf{B},C] & 9,630 &  & [\textbf{B},C] & 9,630 & \\
    &  &            &        & [C,\textbf{B}] &6,000 & 20.7\% & [C,\textbf{B}] &6,000 & 20.7\%  \\
    \hline
Cal & [\textbf{C},B] & 6,000 & & --- &  &&  --- &&  \\
    & [\textbf{C}] & 1,600 & 10.1\% & --- &&  ---&  --- && --- \\
    \hline
Dee & [\textbf{D},A,C] &7,532 &  & [\textbf{D},A,C] & 7,532 &  &  --- && \\
    & [\textbf{D},C] &   846 & 11.1\% & [\textbf{D},C]  &  846 & 11.1\% &  --- && --- \\ \hline
Total & & \bf 75,608 & 100.0\% && 73,738 & 98.6\% & & 73,162 & 96.8\% \\
\hline
\end{tabular}
\caption{IRV election for four candidates showing the elimination of first Cal,
and then Dee, and the final round results.\label{fig:irv}}
\end{figure*}
\end{example}

The second stage in the process is to assign delegates to candidates on the
basis of their tallies.  We first compute, for each viable candidate $c$, the
proportion of the \textit{qualified votes} in their tally, $p_c$. Recall that
we refer to ballots belonging to viable candidates as \textit{qualified} votes.
We denote the number of qualified votes as $Q$.  In the context of IRV, ballots
are qualified if they end up in the tally of a viable candidate.  Non-qualified
ballots result from exhaustion: every candidate in the ballot ranking has been
eliminated (is non-viable).  Where a plurality contest determines viability,
all votes for a viable candidate are qualified.

We denote the set of viable candidates as $\viable$.  Delegates are awarded to viable
candidates as follows:
\begin{enumerate}
\item We compute for each viable candidate $c$ their \textit{delegate quota},
    $q_c = D \times p_c$ where $p_c$ is the proportion of the qualified vote
    given to $c$ (their final tally divided by $Q$). 
\item We assign $i_c = \lfloor q_c \rfloor$ delegates to each candidate $c \in
    \viable$.
\item At this stage, there are $r = D - \sum_{c \in \viable} i_c$ remaining
    delegates to assign. We assign these delegates to the $r$ candidates with
    the largest value of the remainder $q_c - i_c$. One delegate is given to
    each of these $r$ candidates.
\item At this stage, each viable candidate $c$ has received $a_c$ total
    delegates, where $a_c$ is $q_c$ rounded either up or down.
\end{enumerate}

\begin{example}\label{ex:del}
The end result of Examples~\ref{ex:plural} and~\ref{ex:irv} is the same.  The
qualified vote is $Q = 73,162$.  The proportions of the qualified vote in
viable candidates' tallies are: $p_{Ann} = 0.786$; and $p_{Bob} = 0.214$.
Assuming there are $D = 5$ delegates to allocate, we find $q_{Ann} = 3.932$ and
$q_{Bob} = 1.068$.  We initially allocate 3 delegates to Ann and 1 to Bob.  By
comparing the remainders $0.932$ and $0.068$, we allocate the last delegate to
Ann.  So $a_{Ann} = 4$ and $a_{Bob} = 1$.
\end{example}

\section{Auditing Fundamentals}
\label{sec:fund}

A risk-limiting audit is a statistical test of the hypothesis that the reported
outcome is incorrect.  (In the current context, the reported outcome is the
number of delegates finally awarded to each candidate.) If that hypothesis is
not rejected, there is a full hand tabulation, which reveals the true outcome.
If that differs from the reported outcome, it replaces the reported outcome.
The significance level of the test is called the \emph{risk limit}.  A
risk-limiting audit of a trustworthy paper trail of votes limits the risk that
an incorrect electoral outcome will go uncorrected.

Two common building blocks for audits are to compare manual interpretation of
randomly selected ballots or groups of ballots with how the voting system
interpreted them (a \emph{comparison} audit~\cite{stark10}), and to use only
the manual interpretation of the randomly selected ballots (a
\emph{ballot-polling} audit~\cite{lindemanEtal12}).  Ballot polling requires
less infrastructure (some voting systems do not generate or cannot export the
data required for a comparison audit) but generally requires inspecting more
ballots.

Recent work \cite{shangrla} shows that audits of most social choice functions
can be reduced to checking a set of \emph{assertions}.  If all the assertions
are true, the reported election outcome is correct.  Each assertion is checked
by conducting a hypothesis test of its logical negation.  To reject the
hypothesis that the negation is true is to conclude that the assertion is true.
Each hypothesis is tested using a statistic calculated from the audit data.
Larger values of the statistic are unlikely if the corresponding assertion is
false.  If the statistic takes a sufficiently large value, that is statistical
evidence that the assertion is true, because such a large value would be very
unlikely if the assertion were false.  The statistic is generally calibrated to
give \emph{sequentially valid} tests of the assertions, meaning that the sample
of ballots can be expanded at will and the statistic can be recomputed from the
expanded sample, while controlling the probability of erroneously concluding
that the assertion is true if the assertion is in fact false.

The initial sample size is generally chosen so that there is a reasonable
chance that the audit will terminate without examining additional ballots if
the reported results are approximately correct.  If the initial sample does not
give sufficiently strong evidence that all the assertions are correct, the
sample is augmented and the condition is checked again.\footnote{For
sequentially valid test statistics, the sample can be augmented at will; for
other methods, there may be an escalation schedule prescribing a sequence of
sample sizes before conducting a full manual tabulation.} The sample continues
to expand until either all the assertions have been confirmed\footnote{In other
words, the hypothesis that the assertion is false has been rejected at a
sufficiently small significance level.} or the sample contains every ballot,
and the correct result is therefore known.  At any point during the audit, the
auditor can choose to conduct a full manual tabulation.  If the audit leads to
a full manual tabulation, the outcome of that tabulation replaces the reported
outcome if they differ.

The basic assertions for Hamiltonian elections are:
\begin{description}
\item[(Super/sub) majority $p > t$,] where $p$ is the proportion of
    ballots that satisfy some condition (usually the condition is that the
    ballot has a vote for a particular candidate) among ballots that meet some
    validity condition, and $t$ is a proportion in $(0,1]$.
\item[Pairwise majority $p_A > p_B$,] where $p_A$ and $p_B$ are the
    proportions of ballots that meet two mutually exclusive conditions $A$ and
    $B$, among ballots that meet some validity condition.  (Typically, among
    the ballots that contain a valid vote, $A$ is a ballot with a vote for one
    candidate, and $B$ is a ballot with a vote for a different candidate).
\item[Pairwise difference $p_A > p_B + d$,] where $p_A$ and $p_B$ are
    the proportions of ballots that meet two mutually exclusive conditions
    among ballots that meet some validity condition, and $d$ is a constant in
    the range $(-1,1)$.  This is a new form of assertion not previously used,
    that extends pairwise majority assertions.  It is necessary for auditing
    delegate assignment in Hamiltonian elections.
\end{description}

In the SHANGRLA approach to RLAs \cite{shangrla}, each assertion is transformed
into a canonical form: the mean of an \emph{assorter} (which assigns each
ballot a nonnegative, bounded number) is greater than 1/2.  The value the
assorter assigns to a ballot is generally a function of the votes on that
ballot and others and the voting system's interpretation of the votes on that
ballot and others.

For majority assertions, a ballot that satisfies the condition is assigned the
value $1/(2t)$; a valid ballot that does not satisfy the condition is assigned
the value 0; and an invalid ballot is assigned the value 1/2.  For pairwise
majority assertions, a ballot for class A counts as 1 and a ballot for class B
counts as 0.  Ballots that fall outside both classes count as 1/2.

\ignore{
We now construct an assorter for pairwise differences.  Define
$$ f(b) \equiv \left \{ \begin{array}{ll}
   1-d/2, & \mbox{vote for A} \\
   -d/2, & \mbox{vote for B} \\
   1/2-d/2, & \mbox{valid vote in contest but not for A Nor B} \\
   1/2, & \mbox{no valid vote in contest}
   \end{array}
   \right .
$$
$$
   \bar{f} = \frac{n p_A (1-d/2) + n p_B (-d/2) +
                   n (1-p_A-p_B)(1/2-d/2) + (N-n)(1/2)}
                  {N}
$$
$$
   = \frac{n(p_A/2-p_B/2-d/2) + N/2}{N}
$$
$$
   = \frac{n}{2N}(p_A-p_B-d) + 1/2.
$$
This is greater than 1/2 iff $p_A > p_B + d$, but its minimum is $-d /
2$, not $0$, so make the affine transformation
$$
g \equiv \frac{f + d / 2}{1 + d}.
$$
}

For pairwise difference assertions, we define the assorter $g$ which assigns
ballot $b$ the value:
$$ g(b) \equiv \left \{ \begin{array}{ll}
   1/(1 + d),    & \mbox{$b$ has a vote of class A} \\
   0,            & \mbox{$b$ has a vote of class B} \\
   1/(2(1 + d)), & \mbox{$b$ has a valid vote in the contest that is not
                                                                 in A or B} \\
   1/2,          & \mbox{$b$ does not have  a valid vote in the contest.}
   \end{array}
   \right .
$$
Let $\bar{g}$ be the mean of $g$ over the ballots.  We have that $0 \leqslant
g(b) \leqslant 1/(1+d)$, and $\bar{g} > 1/2$ iff $p_A > p_B + d$.  When $d = 0$
this reduces to the pairwise majority assorter if the ``valid'' category is the
same.

The \emph{margin} $m$ of an assertion $a$ is equal to 2 times the mean of its
assorter (when applied to all ballots $\ballots$) minus 1.  An assertion with a
smaller margin will be harder to audit than an assertion with a larger margin.

\subsection{Estimating Sample Size and Risk}

The sample size required to confirm an assertion depends on the sampling design
and the auditing strategy (e.g., sampling individual ballots or batches of
ballots, using ballot polling or comparison); the ``risk-measuring function''
(see \cite{shangrla}); and the accuracy of the tally, among other things.
Because it depends on what the sample reveals, it is random.

There is some flexibility in selecting a set of assertions to confirm IRV
contests \cite{evote2018b}, so the set can be chosen to minimize a measure of
the anticipated workload.  We will estimate the workload on the assumption that
the assertion is true but the reported tallies are not exactly correct.  We
will use the expected sample size as a measure of workload.\footnote{One might
instead seek to minimize a quantile of the sample size or some other function
of the distribution of sample size, for instance, to account for fixed costs
for retrieving and opening a batch of ballots and per-ballot and per-contest
costs.}

Our auditing approach is applicable to any style of auditing.  The workload,
given a set of assertions, varies depending on the style of audit (e.g.,
ballot-level comparison, batch-level comparison, ballot-polling, or a
combination of those) and the sampling design (e.g., with or without
replacement, Bernoulli, stratified or not, weighted or not).  For the purpose
of illustration, in the examples and experiments in this paper, we assume that
the audit will be a ballot-level comparison audit using sampling with
replacement.

Because the sample is drawn with replacement, the same ballot can be drawn more
than once.  Given an assertion $a$, let $ASN(a,\alpha)$ denote the expected
number of draws required to verify $a$ to risk limit $\alpha$, and if
$\mathcal{A}$ is a set of assertions, let $ASN(\mathcal{A},\alpha)$ denote the
expected number of draws required to verify every assertion $a$ in
$\mathcal{A}$ to risk limit $\alpha$.  $ASN$ depends on several factors: the
risk limit $\alpha$; the expected rate of errors (discrepancies) between paper
ballots and their electronic records of various signs and magnitudes (in the
context of comparison auditing); and the margins of the assertions.

We estimate $ASN(\mathcal{A}, \alpha)$ by simulation.  We simulate the sampling
of ballots, one at a time.\footnote{The procedure used to calculate the $ASN$
for an assertion with margin $m$ is available in the public repositories
\url{https://github.com/michelleblom/primaries} and
\url{https://github.com/pbstark/SHANGRLA}.}  An ``overstatement'' error  is
introduced with a pre-specified probability $e$.  If the sample reveals one or
more overstatements, the measured risk (i.e., the $P$-value of the hypothesis
that the assertion is false) increases by an amount that depends on margin $m$.
Otherwise, the measured risk decreases by an amount that depends on $m$.  We
continue to sample ballots until the measured risk falls below $\alpha$ or
until every ballot has been manually reviewed, in which case the outcome based
on the manual interpretations replaces the original reported results.  We take
the median of the number of ballots sampled over $N$ simulations as an estimate
of $ASN(\mathcal{A}, \alpha)$.  Inaccuracy of this estimate affects whether the
selected assertions result in the smallest expected workload, but does not
affect the risk limit.  For the examples and experiments in this paper, we use
$e = 0.002$ (equivalent to 2 errors per 1,000 ballots), $N = 20$, and a risk
limit of 5\%.

\ignore{
For the types of assertions listed above, for an unstratified sample of
individual ballots, and on the assumption that the sample size remains a small
fraction of the total number of ballots, the $ASN$ scales like the reciprocal
of the ``assorter margin'' for ballot-level comparison audits and like the
square of the reciprocal of the ``assorter margin'' for ballot-polling audits.
}

\section{Auditing Viability}
\label{sec:eligible}

The first stage of the election identifies the viable candidates.  We introduce
notation for the assertions we will use to audit viability, as follows:
\begin{itemize}
\item $\Viable(c, E, t)$: Candidate $c$ has at least proportion $t$ of the vote
    after the candidates in set $E$ have been eliminated.  This amounts to a
    simple majority assertion $p_c > t$ after candidates in $E$ are
    eliminated.
\item $\NonViable(c, E, t)$: Candidate $c$ has less than proportion $t$ of the
    vote after candidates $E$ have been eliminated. This amounts to a simple
    majority assertion $\bar{p}_c > 1-t$ where $\bar{p}_c$ is the
    proportion of valid votes for candidates other than $c$ after candidates
    $E$ are eliminated.
\item $\mbox{\emph{IRV}}(c, c', E)$: Candidate $c$ has more votes than
    candidate $c'$ after candidates $E$ have been eliminated. This amounts to a
    pairwise majority assertion.
\end{itemize}
If the first stage is a plurality vote, $E \equiv \emptyset$: the elimination
in the first stage only occurs for $\mbox{\emph{IRV}}$.

Consider an election $\election = \langle \cands, \ballots, \threshold \rangle$
with candidates $\cands$, cast ballots $\ballots$, and viability threshold
$\threshold$ ($\threshold = 0.15$ for the primary elections we will examine).
The outcome of this election is a set of viable candidates, $\viable \subseteq
\cands$, together with, in the case of instant runoff Hamiltonian elections, a
sequence of eliminated candidates, $\pi$.  To check that the set of candidates
reported to be viable really are the viable candidates, we test assertions that
rule out all all other possibilities.  Consider the subset $\viable' \subseteq
\cands$, where $\viable' \neq \viable$. We can demonstrate that $\viable'$ is
not the true set of viable candidates by showing that some candidate $c \in
\viable'$ does not belong there.  \ignore{-- i.e., they do not have at least
    $\threshold$\% of the vote when all candidates not in $\viable'$ have been
eliminated. } We can also rule out $\viable'$ as an outcome by showing that
there is a candidate $c \notin \viable'$ that does in fact belong in the viable
set.  We aim to find the `least effort' set of assertions $\asserts$ that, if
shown to hold in a risk-limiting audit, confirm that (i) each candidate in
$\viable$ is viable, and (ii) no candidate $c' \notin \viable$ is viable.

\subsection{Viability: Plurality Hamiltonian Elections}

For each viable candidate $v \in \viable$ we need to verify the assertion
$\Viable(v,\emptyset,\threshold)$.  For each non-viable candidate $n \in
\cands\setminus\viable$ we need to verify the assertion
$\NonViable(n,\emptyset,\threshold)$.  Let $\asserts$ be the union of these two
sets of assertions.  Note that $\asserts$ rules out any other set of viable
candidates $\viable' \neq \viable$.  

\begin{example}
To audit the first stage of the election
of Example~\ref{ex:plural}, we verify the assertions
$\asserts = \{
\Viable(\text{Ann},\emptyset,0.15)$,
$\Viable(\text{Bob},\emptyset,0.15)$,
$\NonViable(\text{Cal},\emptyset,0.15)$,
$\NonViable(\text{Dee},\emptyset,0.15)\}$.
The margins associated with these assertions are 4.073, 0.378, 0.152, and
0.163, respectively.  The expected number of ballots we need to compare to the
corresponding cast vote records to audit these assertions, assuming an
overstatement error rate of 0.002 and a risk limit of $\alpha = 5\%$, are,
respectively 1, 17, 46, and 42.  The overall ASN for the audit is 46 ballots.
\end{example}

\ignore{
Ann: 57532, share = 1/0.30, margin= 4.073
Bob: 15630, margin = 0.378
Cal: 1600, margin = 0.152
Dee: 846, margin = 0.163
Total: 75608
ex: 2446
}

\subsection{Viability: Instant-Runoff Hamiltonian Elections}

Efficient RLAs for IRV have been devised only recently \cite{evote2018b}.  To
audit the first stage of an IRV Hamiltonian election we must eliminate the
possibility that a different set of candidates is viable.  This means that we
need to look at every other set of candidates, and propose an assertion that
will show that set is not viable. 

In contrast to auditing a simple IRV election, where there are $|\cands| - 1$
potential winners other than the reported winner, a Hamiltonian election
typically has many more.  Let $M = \lfloor 1/\threshold \rfloor$ be the maximum
possible number of viable candidates.  The number of possible winner sets $O$
is
$$
O \equiv \binom{|\cands|}{M} + \binom{|\cands|}{M-1} + \cdots
  + \binom{|\cands|}{1}.
$$

We can show that a subset of candidates $\viable'$ is not the set of viable
candidates in a number of ways:
\begin{itemize}
\item we could show that the tally of at least one $c \in \viable'$ does not
    reach the required threshold assuming all candidates not mentioned in
    $\viable'$ have been eliminated
\item we could show that there is a candidate $c' \notin \viable'$ that is
    viable on the basis of their first preferences, so any potential set of
    viable candidates must include $c'$
\item we could show that the unmentioned candidates could not have been
    eliminated in a sequence that would result in $\viable'$.
\end{itemize}

\subsubsection{Reducing the set of subsets}

While there are many possible alternate winner sets $\viable'$, we can rule out
many of these easily.  We examine the assertions $\Viable(w, \emptyset,
\threshold)$ for any candidate $w \in \cands$ who had more than the proportion
$\threshold$ of the vote initially.  This assertion will be easy to verify, as
long as the proportion is not too close to $\threshold$.  This assertion rules
out any subset $\viable'$ where $w \not \in \viable'$.  Let $\winners$ be the
set of candidates where this assertion is expected to hold.

Next we examine the assertions $\NonViable(l, \cands \setminus \winners
\setminus \{l\}, \threshold)$ for those candidates $l$ who are not mentioned in
at least $\threshold$ of the ballots, when all but the definite winners
$\winners$ and $l$ are eliminated.  In this case candidate $l$ can never reach
$\threshold$ proportion of the votes.  Again this assertion is easy to verify
as long as the proportion of such votes is not close to $\threshold$.  This
assertion removes any subset $\viable'$ where $l \in \viable'$.  Let $\losers$
be the set of candidates where this assertion is expected to hold.

We collect together $\asserts = \{ \Viable(w, \emptyset, \threshold) \mid w \in
\winners \} \cup \{ \NonViable(l, \cands \setminus \winners \setminus \{l\},
\threshold) \mid l \in \losers\}$.  If these assertions hold, we only need to
consider subsets of viable candidates $\textbf{V} = \{ \viable' \subseteq
\cands \mid \winners \subseteq \viable', \viable' \cap \losers = \emptyset\}
\setminus \{ \viable\}$.  There are only $|\textbf{V}|$ subsets to further
examine, where
$$
|\textbf{V}| = \binom{|\cands\setminus \winners \setminus \losers|}
                     {M - |\winners|} +
               \cdots +
               \binom{|\cands\setminus \winners \setminus\losers|}{1} - 1
$$

\subsubsection{Selecting assertions for the remaining subsets}

We now need to select a set of assertions that rule out any alternate set of
viable candidates $\viable' \in \textbf{V}$.  To form these assertions, we
visualise the space of alternate election outcomes as a tree.  We use a
branch-and-bound algorithm to find a set of assertions that, if true, will
prune (invalidate) all branches of this tree.   At the top level of this tree
is a node for each possible $\viable' \in \textbf{V}$.  Each node defines an
(initially empty) sequence of candidate eliminations, $\pi$, and a set of
viable candidates, $\viable'$. These nodes form a frontier, $\frontier$.

Our algorithm maintains a lower bound $\text{LB}$ on the estimated auditing
effort (EAE) required to invalidate all alternate election outcomes, initially
setting $\text{LB} = 0$.  For each node $n = (\emptyset, \viable')$ in
$\frontier$, we consider the set of assertions that could invalidate the
outcome that it represents.  Two kinds of assertion are considered at this
point:
\begin{itemize}
\item $\Viable(c', \losers, t)$ for each candidate $c' \in \cands$ that does
    not appear in $\viable'$, and whose first preference tally exceeds $t$
    proportion of the vote when only candidates in $\losers$ are eliminated;
\item $\NonViable(c, \cands\setminus\viable', t)$ for each candidate $c \in
    \viable'$ whose tally, if all candidates $c'\in \cands\setminus\viable'$
    have been eliminated, falls below $t$ proportion of the vote.
\end{itemize}
We assign to $n$ the assertion $a$ from this set with the smallest EAE
(EAE[$n$] = $\mbox{\emph{ASN}}(\{a\},\alpha)$ where we use the method for
estimating ASN previously described.  If no such assertion  can be formed for
$n$, we give $n$ an EAE of $\infty$, EAE[$n$] = $\infty$.  We then select the
node in $\frontier$ with the highest EAE to expand.

To expand a node $n = (\pi, \viable')$, we consider the set of candidates in
$\cands$ that do not currently appear in  $\pi$ or $\viable'$.  We denote this
set of `unmentioned' candidates, $\unmentioned$. For each candidate $c' \in
\unmentioned$, we form a child of $n$ in which $c'$ is appended to the front of
$\pi$.  For instance, the node $([c'], \viable')$ represents an outcome in
which $c'$ is the last candidate to be eliminated, after which all remaining
candidates, $c \in \viable'$, have at least $t = T$ proportion of the cast
votes.  All unmentioned candidates are assumed to have been eliminated, in some
order, before $c'$.  For each newly created node, we look for an assertion that
could invalidate the corresponding outcome.  Two kinds of assertion are
considered to rule out an outcome $n' = ([c'|\pi'], \viable')$:
\begin{itemize}
\item $\Viable(c', \unmentioned\setminus\{c'\},t)$ for each candidate $c' \in
    U$ that has at least $t$ proportion of the vote in the context where
    candidates $\unmentioned\setminus\{c'\}$ have been eliminated. Candidate
    $c'$ thus cannot have been eliminated at this point;

\item $\mbox{\emph{IRV}}(c', c, \unmentioned\setminus\{c'\})$ for each
    candidate $c' \in U$ that has a higher tally than some candidate $c \in
    \pi' \cup \viable'$ in the context where candidates
    $\unmentioned\setminus\{c'\}$ have been eliminated. Candidate $c'$ thus
    cannot have been eliminated at this point.
\end{itemize}
We assign to each child of $n$ the assertion $a$ from this set with the
smallest $ASN(\{a\},\alpha)$, and replace $n$ on our frontier with its
children. If neither of the above two types of assertion can be created for a
given child of $n$, the child is labelled with an EAE of $\infty$ (EAE[$n'$] =
$\infty$).  We continue to expand nodes in this fashion until we reach a leaf
node, $l = (\pi, \viable')$, where $\pi \cup \viable' = \cands$ (all candidates
are mentioned either in the elimination sequence $\pi$ or in the viable set
$\viable'$).  We assign to $l$ an invalidating assertion of the above two
kinds, if possible. We consider all the nodes in the branch that $l$ sits on,
and select the node $n_b$ associated with the least cost assertion $a$. We add
$a$ to our set of assertions to audit $\asserts$, prune $n_b$ and all of its
descendants from the tree, and  update our lower bound on audit cost
$\text{LB}$ to max($\text{LB}$, EAE[$a$]). We then look at all nodes on our
frontier that can be pruned with an assertion that has an EAE $\leqslant
\text{LB}$.  We add those assertions to $\asserts$, and prune the nodes from
the frontier. The algorithm stops when the frontier is empty. If we discover a
branch whose best assertion has an EAE of $\infty$, the algorithm stops in
failure---indicating that a full manual count of the election is required.

This branch-and-bound algorithm is a variation of that described by
\cite{evote2018b,RAIRE} for generating an audit specification for an IRV
election. It has been altered for the context where the ultimate outcome is a
set of winning candidates---the viable candidates---and not one winner, left
standing after all others are eliminated.

\section{Auditing Delegate Assignment}
\label{sec:delegate}

The Hamilton method for proportional representation is used to assign delegates
to viable candidates.  It might appear that auditing the Hamilton method
requires checking some delicate results, for instance, whether candidate $A$
received at least 2 delegate quotas when  candidate $A$ actually received 2.001
delegate quotas.  However, this is not necessary, because candidate $A$ can
receive 2 delegates without having at least 2 delegate quotas.  For example, if
$A$ receives 1.999 quotas $A$ may still end up with 2 delegates.  Our auditing
method avoids checking such things.

The audit instead examines all pairs of viable candidates, including those
receiving no delegates.  For each pair of viable candidates $n$ and $m$ we
check whether $(q_n - (a_n - 1)) < 1 + (q_m - (a_m - 1))$ which
requires that the quota of $n$ is not 1 more than the quota for $m$, after
removing all received delegates but the last.  This can be equivalently
rewritten as
\begin{equation}
p_m > p_n + \frac{a_m - a_n - 1}{D}, \quad n,m \in \viable, n \neq m.
\label{as:comp}
\end{equation}
In the case that $q_m$ was rounded up and $q_n$ was rounded down, this
captures that the remainder for $m$ was greater than the remainder for $n$:
$p_m D - (a_m - 1) > p_n D - a_n$.

We show that if the delegates  are wrong with respect to the true votes, then
one of these assertions is violated.

\begin{theorem}
Suppose the number of assigned delegates $a_c$ to each viable delegate $c$ is
incorrect, then one of the assertions of Equation~\ref{as:comp} will be
violated.
\end{theorem}
\begin{proof}
Suppose $a'_c$ is the \emph{true} number of delegates that should have been
awarded to each candidate $c$.  Since $\sum_{c \in \viable} a_c$ $= D$ and
$\sum_{c \in \viable} a'_c = D$, and they differ, there must be at least one
candidate $m \in \viable$, where $a_m \geqslant a'_m + 1$, who was awarded too
many delegates, and at least one $n \in \viable$, where $a_n \leqslant a'_n -
1$, who was awarded too few.

Since $a'_m$ is the true number of delegates awarded to $m$ we know that the
(\emph{true}) proportion of the vote for $m$, $p_m$, must be (a) $p_m D < a'_m
\leqslant a_m-1$ if $m$ was rounded up or (b) $p_mD < a'_m +1 \leqslant a_m$ if
$m$ was rounded down.  Similarly, since $a'_n$ is the true number of delegates
awarded to $n$ we know that either (c) $p_n D \geqslant a'_n - 1 \geqslant a_n$
if $n$ was rounded up, or (d) $p_n D \geqslant a'_n \geqslant a_n + 1$ if $n$
was rounded down.

If we add these two inequalities for combinations (a)+(c) or (b)+(d) we get
$p_m D + a_n < p_nD + a_m - 1$.
For the combination (a)+(d) we get
$p_m D + a_n + 1 < p_n D + a_m- 1$.
Any of these cause the assertion
$p_m > p_n + \frac{a_m- a_n -1}{D}$ to be falsified.
For the last case (b)+(c) we need a stricter comparison, which we obtain by
comparing the remainders.  Since $m$ was rounded down and $n$ was rounded up,
we know that remainder for $m$ was less than the remainder for $n$, i.e.,
$p_m D - a'_m < p_n D - (a'_n - 1)$. Hence
$p_m D < p_n D + a'_m - (a'_n - 1) \leqslant p_n D + (a_m - 1) - a_n$.
Again the assertion
$p_m > p_n + \frac{a_m - a_n - 1}{D}$ is falsified.
\qed
\end{proof}

\begin{example}
Consider the delegate allocation of Example~\ref{ex:del}.  Recall that the
proportions of the qualified vote are $p_{\text{Ann}} = 0.7836$ and
$p_{\text{Bob}} = 0.2136$.  We audit that $p_{\text{Ann}} < p_{\text{Bob}} +
4/5$ and $p_\text{Bob} < p_{\text{Ann}} - 2/5$.  These facts require much less
work to prove, than for example auditing that $p_\text{Bob} > 1/5$. The margins
associated with the above pairwise difference assertions are 1.1 and 0.12,
respectively.  Assuming an error rate of 0.002, and a risk limit of $\alpha =
5\%$, the ASNs associated with these assertions are 5, and 59, ballots.
\end{example}

\section{Experiments}
\label{sec:exp}

We consider the set of Hamiltonian elections conducted as part of the selection
process for the 2020 Democratic National Convention (DNC) presidential nominee.
Most of these primaries determine candidate viability via a plurality vote.
Several states, including Wyoming and Alaska, use IRV.  We estimate the number
of ballots we would need to check in a comparison audit of these primaries.
For each of these  primaries, we audit the viability of candidates on the basis
of the statewide vote, and that each  viable candidate deserved the delegates
that were awarded to them. We consider only the delegates that are awarded on
the basis of statewide vote totals (PLEO\footnote{Party Leaders and Elected
Officials} and at-large) as these are readily available.\footnote{Data for
plurality-based primaries was obtained from www.thegreenpapers.com/P20.  Data
for  IRV-based primaries we consider was provided by the relevant state-level
Democrats.} In each proportional DNC primary, viable candidates must attain at
least 15\% of the total votes cast. 

The full code used to generate the assertions for each DNC primary, and
estimate the ASN for each audit, is located at:
\begin{center}
\url{https://github.com/michelleblom/primaries}
\end{center}

Table \ref{tab:PluralityStates} reports the expected number of ballot samples
required to perform three levels of audit in each plurality and IRV-based
primary conducted for the 2020 DNC.\footnote{A small number of DNC 2020
primaries that did not use proportional allocation of delegates were not
considered, in addition to those for which we could not obtain data.} Level~1
checks only that the reportedly viable candidates have at least 15\% of the
vote, and all other candidates do not. Level~2 checks candidate viability and
that  each viable candidate $c$, with $a_c$ allocated delegates, deserved at
least $a_c - 1$ of them.  We introduce this level because, as the table shows,
sometimes the complete auditing of the final delegate counts is difficult.
Level~3 checks candidate viability and that each viable candidate deserved all
of their allocated delegates. The assertions required to check the allocation
of a candidate's final delegate are the hardest to audit.

Of the primaries in Table \ref{tab:PluralityStates}, Maine (ME), New Hampshire
(NH), Washington (WA), Texas (TX), Idaho (ID), Massachusetts (MA), California
(CA), and Minnesota (MN), were considered to be close with differences of less
than 10\% in the statewide vote between the two leading candidates. In Maine,
the difference in the statewide vote for Biden and Sanders was less than 1\% of
the cast vote. Auditing the Maine primary, however, is expected to require only
a sample of 189 ballots.  The Rhode Island (RI) primary, in contrast, requires
a full manual recount.  In RI, Sanders narrowly falls below the 15\% threshold
with 14.93\% of the vote to  Biden's 76.67\%.  The margins that determine the
complexity of these audits are the extent to which a candidates' vote falls
below or exceeds the relevant threshold, and the relative size of the
remainders in candidates' delegate quotas as a proportion of the number of
delegates available.

\begin{table*}[t]
\centering
\caption{Estimated sample size required to audit viability and delegate
distribution (PLEO and at-large) in all proportional (plurality or IRV-based)
DNC primaries in 2020 for which data was available. Levels 1, 2, and 3 audit
candidate viability, that each viable candidate deserved almost all of their
allocated delegates, and that they deserved all of their delegates,
respectively. An error rate of 0.002 (an expectation of 1 error per 1,000
ballots) was used in the estimation of sample sizes. The symbol `--' indicates
that a full recount is required. The number of candidates ($|\cands|$) and
total number of cast ballots ($|\ballots|$) is stated for each election.}
\begin{tabular}{|c|r|r|r|r|r||c|r|r|r|r|r|}
\multicolumn{12}{l}{Plurality-based Primaries} \\
\hline
       &  & & \multicolumn{3}{c||}{ASN ($\alpha = $ 5\%)} &    &    & &
              \multicolumn{3}{c|}{ ASN ($\alpha = $ 5\%)} \\
\cline{4-6} \cline{10-12}
State  & $|\cands|$ & $|\ballots|$ & Level 1 & Level 2 & Level 3    &
State  & $|\cands|$ & $|\ballots|$ & Level 1 & Level 2 & Level 3 \\
\hline
AL & 15&   452,093 & 182 & 182 & 1,352 &  NC & 16& 1,332,382& 350& 350 &808  \\
AR & 18&   229,122 & 121 & 121 & 1,154 &  NE & 4 & 164,582&925 & 925 & 925 \\
AZ & 12&   536,509 & 71  & 71 & 120   & NH  & 34& 298,377 & 104&  104 &  155\\
CA & 21& 5,784,364 & 395 & 1,258 &  3,187,080  & NJ& 3  &958,202 &4,514 & 4,515 & 4,514\\
CO & 13&   960,128 & 42& 42 & --  &  NM & 7 & 247,880 &1,812& 1,812 & 1,812\\
CT &  4&   260,750 & 174& 174 & 174  & NY & 11 & 752,515 & 56& 731 & 486,495\\
DC &  5&   110,688 & 334& 334& 334  &  OH & 11 & 894,383& 61& 334 & -- \\
DE &  3&    91,682 & 80& 80 & 80  &  OK & 14 & 304,281& 649& 649 & 649\\
FL & 16& 1,739,214 & 91& 208 & 766 & OR & 5  & 618,711& 111& 111 & 191\\
GA & 12& 1,086,729 & 107&218 & 218 & PA &  3 & 1,595,508& 48& 167 & 642 \\
ID & 14& 1,323,509 &  143&  143&  143 & PR & 11  & 7,022& 412 & 412 & 412 \\
IL & 12& 1,674,133 & 44& 140 & 620 &  RI  & 7& 103,982&-- & --& -- \\
IN &  9&   474,800 & 391& 391 & 391 & SC & 12 &539,263 & 165& 165 &34,546  \\
KY & 11&   537,905 & 209& 209 & 209 &  SD & 2 & 52,661& 13& 13 & 216 \\
LA & 14&   267,286 & 79& 98& 98 &  TN  & 16& 516,250 &235 & 235 & 1,203 \\
MA & 18& 1,417,498 & 185&  185 & 832&   TX &17 & 2,094,428 &  1,282 &  1,282 &  2,133  \\
MD & 15& 1,050,773 & 83& 170 & 170  &  UT  & 15& 220,582 &262 &  262 &781  \\
ME & 13&   205,937 & 189  &  189 &  189   & VA  & 14 &  1,323,509& 143& 204 & 1,309 \\
MI & 16& 1,587,679 & 57&  118 &  --  &VT  & 17& 158,032 & 289& 289 & 508 \\
MN & 16&   744,198 & 309& 309 &  6,195 &   WA & 15& 1,558,776& 103& 127 &  617 \\
MO & 23&   666,112 & 44& 130 & --  &WI  &14 & 925,065 & 44& 144 & 878 \\
MS & 10&   274,391 & --& -- & --  & WV  & 12& 187,482& 213& 213 & 213 \\
\cline{7-12}
MT  & 4&   149,973 & 5,159& 5,159 & 5,159 &\multicolumn{6}{c|}{} \\
\hline
\multicolumn{12}{l}{} \\
    \multicolumn{12}{l}{IRV-based Primaries} \\
    \hline
AK & 9 & 19,811&  88 & 88 & 88   & WY & 9 & 15,428 & 66 & 87 & 452 \\
\hline
\end{tabular}
\label{tab:PluralityStates}
\end{table*}

Table \ref{tab:CompareHardness} contrasts several of the hardest primaries of
Table \ref{tab:PluralityStates} to audit with some of the easiest. We record,
for each of these primaries: the number of at-large delegates being awarded;
the delegate quotas computed for each viable candidate; the difference between
the decimal portion of these quotas (the remainder) divided by the number of
available delegates; and the estimated auditing effort (ASN) for the primary.
For the first four primaries in the table, the last awarded at-large delegate
is the hardest to audit.

The use of IRV for determining candidate viability does not make a Hamiltonian
election more difficult to audit. While more assertions are created to audit an
IRV-based primary, the difficulty of any audit is based on the cost (ballot
samples required) of its most expensive assertion. Since all assertions are
tested on each ballot examined, the principle cost is retrieving the correct
ballot.  The audit specifications generated for the Wyoming and Alaskan
primaries contain 78 and 89 assertions, respectively. The number of assertions
formed for a plurality-based primary is proportional to the number of
candidates. NH, involving the most candidates at 34, has 48 assertions to
audit.

\begin{table}
\centering
\caption{Hard (top) and relatively easy (bottom) primaries for which to audit
the last assigned at-large delegate to each candidate. The number of at-large
delegates $D$; the delegate quotas for Biden and Sanders; and the difference
between the remainder of their quotas (divided by $D$) is reported, since this
corresponds to the tightness of equation (1).}
\begin{tabular}{|c|c|c|c|c|c|}
\hline
      &    & \multicolumn{2}{c|}{Quotas} & Rem. & \\
      \cline{3-4}
State & $D$  &  Biden & Sanders & Diff. / $D$ & ASN \\
\hline
CA & 90 & 50.688 & 39.312 & 0.004 & 3.2$\times 10^6$\\
MO & 15 &  9.524 &  5.476 & 0.003 & -- \\
NY & 61 & 47.629 & 13.371 & 0.004 & 486,495 \\
SC & 12 &  8.533 &  3.467 & 0.006 &  34,546 \\
\hline
ME &  5 &  2.050 &  1.993 & 0.19  & 189 \\
AZ & 14 &  8.010 &  5.990 & 0.07  & 120 \\
OR & 13 &  9.948 &  3.052 & 0.07  & 191 \\
\hline
\end{tabular}
\label{tab:CompareHardness}
\end{table}

The computational cost of generating these audit specifications is not
significant.  On a machine with an Intel Xeon Platinum 8176 chip (2.1GHz), and
1TB of RAM, the generation of an audit specification for Wyoming and Alaska
takes 0.3s and 0.4s, respectively.  The time required to generate an audit for
each of the plurality-based primaries in Table \ref{tab:PluralityStates} ranges
from 0.2ms to 0.24s (and 0.03s on average).

\section{Conclusion}
\label{sec:conc}

We provide an effective method for auditing delegate allocation by proportional
representation (the Hamilton method), the first we know of for elections of
this kind.  Usually the audit only requires examining a small number of
ballots.  This could be used for primary elections in the USA and other
elections in Russia, Ukraine, Tunisia, Taiwan, Namibia and Hong Kong.

We provide a version suitable for Democratic primaries in Alaska, Hawaii,
Kansas, and Wyoming, which use a modified form where viability is decided using
IRV.

To audit these elections we defined a new assertion for pairwise differences
and corresponding assorter, which may be useful for auditing other methods.

\afterpage{\clearpage}

\bibliographystyle{splncs04}
\bibliography{BIB,extra}

\begin{thebibliography}{10}
\providecommand{\url}[1]{\texttt{#1}}
\providecommand{\urlprefix}{URL }
\providecommand{\doi}[1]{https://doi.org/#1}

\bibitem{evote2018b}
Blom, M., Stuckey, P.J., Teague, V.: Ballot-polling risk limiting audits for
  {IRV} elections. In: Krimmer, R., Volkamer, M., Cortier, V., Goré, R.,
  Hapsara, M., Duenas-Cid, U.S.D. (eds.) Proceedings of the E-Vote-ID 2018:
  Third International Joint Conference on Electronic Voting. LNCS, vol. 11143,
  pp. 17--34. Springer (2018)

\bibitem{blom18}
Blom, M., Stuckey, P.J., Teague, V.: Computing the margin of victory in
  preferential parliamentary elections. In: Proceedings of the E-Vote-ID 2018:
  Third International Joint Conference on Electronic Voting. LNCS, vol. 11143,
  pp. 1--16. Springer (2018)

\bibitem{blom16}
Blom, M., Teague, V., Stuckey, P.J., Tidhar, R.: Efficient computation of exact
  {IRV} margins. In: European Conference on Artificial Intelligence (ECAI). pp.
  480--488 (2016)

\bibitem{RAIRE}
Blom, M.L., Stuckey, P.J., Teague, V.: Risk-limiting audits for {IRV}
  elections. CoRR  \textbf{abs/1903.08804} (2019),
  \url{http://arxiv.org/abs/1903.08804}

\bibitem{kroll2014}
Kroll, J.A., Halderman, J.A., Felten, E.W.: Efficiently auditing multi-level
  elections. In: Krimmer, R., Volkamer, M. (eds.) Proceedings of Electronic
  Voting 2014 (EVOTE2014). pp. 93--101. TUT Press (2014)

\bibitem{lindemanStark12}
Lindeman, M., Stark, P.: A gentle introduction to risk-limiting audits. IEEE
  Security and Privacy  \textbf{10},  42--49 (2012)

\bibitem{lindemanEtal12}
Lindeman, M., Stark, P., Yates, V.: {BRAVO}: Ballot-polling risk-limiting
  audits to verify outcomes. In: Proceedings of the 2011 Electronic Voting
  Technology Workshop / Workshop on Trustworthy Elections (EVT/WOTE '11).
  {USENIX} (2012)

\bibitem{sarwate2013risk}
Sarwate, A., Checkoway, S., Shacham, H.: Risk-limiting audits and the margin of
  victory in nonplurality elections. Politics, and Policy  \textbf{3}(3),
  29--64 (2013)

\bibitem{stark2008conservative}
Stark, P.: Conservative statistical post-election audits. Annals of Applied
  Statistics  (2008)

\bibitem{stark10}
Stark, P.: Super-simple simultaneous single-ballot risk-limiting audits. In:
  Proceedings of the 2010 Electronic Voting Technology Workshop / Workshop on
  Trustworthy Elections (EVT/WOTE '10). {USENIX} (2010)

\bibitem{shangrla}
Stark, P.B.: Sets of half-average nulls generate risk-limiting audits:
  Shangrla. In: Bernhard, M., Bracciali, A., Camp, L.J., Matsuo, S., Maurushat,
  A., R{\o}nne, P.B., Sala, M. (eds.) Financial Cryptography and Data Security.
  pp. 319--336. Springer International Publishing, Cham (2020)

\bibitem{StarkTeague2014}
Stark, P.B., Teague, V.: Verifiable {E}uropean elections: Risk-limiting audits
  for {D}'{H}ondt and its relatives. {USENIX} Journal of Election Technology
  and Systems ({JETS})  \textbf{1}(3),  18--39 (Dec 2014),
  \url{https://www.usenix.org/jets/issues/0301/stark}

\end{thebibliography}

\end{document}